\documentclass{article}
\usepackage{mst-stylefile}
\usepackage{amsmath}
\usepackage{amssymb}
\usepackage{tikz}
\usetikzlibrary{arrows}

\def\ci{\perp\!\!\!\perp}

\begin{document}

\title{A Rigorous Analysis of the Clauser-Horne-Shimony-Holt Inequality Experiment When Trials Need Not Be Independent\thanks{This work was partially supported by grant FA9550-13-1-0135 from the US Air Force Office of Scientific Research and grant N00014-10-1-0329 P00004 from the US Office of Naval Research.}}



\author{Peter Bierhorst\\Tulane University}




\date{}

\maketitle

\begin{abstract}
The Clauser-Horne-Shimony-Holt (CHSH) inequality is a constraint that local hidden variable theories must obey. Quantum Mechanics predicts a violation of this inequality in certain experimental settings. Treatments of this subject frequently make simplifying assumptions about the probability spaces available to a local hidden variable theory, such as assuming the state of the system is a discrete or absolutely continuous random variable, or assuming that repeated experimental trials are independent and identically distributed. In this paper, we do two things: first, show that the CHSH inequality holds even for completely general state variables in the measure-theoretic setting, and second, demonstrate how to drop the assumption of independence of subsequent trials while still being able to perform a hypothesis test that will distinguish Quantum Mechanics from local theories. The statistical strength of such a test is computed.
\end{abstract}

\newpage

\section{Introduction}\label{s:introduction}

It has been known since the 1964 publication by J. Bell \cite{bell:1964} that Quantum Mechanics makes predictions incompatible with any so-called local hidden variable theory (LHVT). The conflict can be tested experimentally with an instrument that generates entangled particles and two particle detectors that can measure certain properties, such as spin or polarization. For such an experiment, the Clauser-Horne-Shimony-Holt (CHSH) inequality \cite{chsh:1969} provides a constraint on the possible outcomes under a LHVT; according to the prediction of Quantum Mechanics, the constraint will be violated. The profound physical implications of the CHSH experiment have been long discussed, and more recently, the experiment has been found to have new applications in the field of device-independent quantum key distribution \cite{BHK:2005,acin:2007} and device-independent randomness expansion \cite{pironio:2010}.

The CHSH inequality is as follows:
\begin{equation}\label{e:gCHSH}
-2 \le  E_{ab}(D_1 D_2) - E_{a'b}(D_1 D_2) + E_{ab'}(D_1 D_2) + E_{a'b'}(D_1 D_2) \le 2,
\end{equation}
where $E_{xy}(D_1D_2)$ is the expected value of a detectable quantity $D_1D_2$ when the measurement apparatus has setting $xy$. The precise meaning of this will be explained in the next section, but what is immediately clear is that (\ref{e:gCHSH}) is a \emph{probabilistic} statement, asserting that under locality, a particular function of the probabilities of various experimental outcomes cannot exceed a certain quantity. According to the predictions of Quantum Mechanics, this quantity will be exceeded. 

The probabilistic nature of the constraint (\ref{e:gCHSH}) raises two issues. The first issue is: how does one build an appropriate mathematical model for the experiment? In the original proofs of the Bell \cite{bell:1964} and CHSH \cite{chsh:1969} inequalities, it is tacitly assumed that the random variable that models the state of the system can be taken to be absolutely continuous, in the sense that it has a probability density function. Though this is a fairly reasonable assumption to make about a random variable modeling a real-world phenomenon, in the interest of full generality it would be best to not make such a claim. In some recent work on hidden variable models \cite{fritz:2012,brandenburger:2012}, authors have worked in a more general measure-theoretic setting, though the frameworks set out in \cite{fritz:2012,brandenburger:2012} have not been used to prove the original CHSH inequality or model repeated trials of the experiment.

The second issue is: how does one draw a conclusion from the experimental data? As the constraints on LHVTs are probabilistic, any single execution of the experiment does not provide evidence for or against any one particular theory. (This is for the same reason that the result of a single coin toss does not tell you if a coin is biased.) The standard strategy for dealing with this is to run many trials of the experiment and compare the sample means to the predicted expectations. There is a problem, though -- the sample means needn't converge to the predicted expectations. One could expect convergence if one could assume that subsequent trials are independent and identically distributed (i.i.d.) -- but plausible though this assumption seems, it need \emph{not} be satisfied by a LHVT. Indeed, it is not hard to devise a mechanism for a LHVT to violate this assumption: detected particles could leave some sort of residue in the particle detectors that biases the outcome of the next incoming particle. This complication has been referred to as the ``memory loophole" in \cite{barrett:2002}, and it has also been addressed in \cite{gill:2003}. (Possible interdependence between experimental trials can also cause security problems for quantum key distribution protocols, as seen in \cite{hanggirenner:2013,barrett:2013}.) \cite{barrett:2002} concludes that, even allowing for time dependence, quantum mechanical experimental data can be reliably distinguished from the data produced by any LHVT; however, the paper uses some informal justifications and assumes that the state random variable is absolutely continuous. \cite{gill:2003} reaches the same conclusion with more rigor, but the exact bound on the statistical p-value derived from the Azuma-Hoeffding inequality \cite{hoeffding:1963,azuma:1967} can be improved on. (Here, the ``p-value" is the probability of seeing data as or more extreme than what is observed experimentally, under a LHVT.) 

In this paper, we resolve these two issues simultaneously. We present a completely general measure-theoretic model for the Bell test experiment, making no unnecessary assumptions about the random variables involved. Using this framework, we show that the CHSH inequality can still be derived. The framework can be extended in a natural way to accommodate repeated trials that need not be independent and/or identically distributed. In the extended framework, we prove that a hypothesis test can reliably distinguish between Quantum Mechanics and LHVTs, where the null hypothesis is that nature is governed by a LHVT. Interestingly, the p-value for rejecting the null hypothesis is shown to be the same as it would be if we restricted the null hypothesis to the narrower class of LHVTs that are i.i.d. That is, allowing for LHVTs with memory does not increase the probability of violating the CHSH inequality under the null hypothesis. The calculated p-value of the hypothesis test described in this paper compares favorably to other calculations of p-values in Bell-inequality experiments \cite{gillvandam:2005,zhang:2011}.

The paper uses the formalism of measure-theoretic probability (see, e.g., \cite{chung:1974}). The structure is as follows: in Section \ref{s:setting}, we describe the mathematical model for the CHSH experiment, in Section \ref{s:math_develop}, we derive the CHSH inequality in this setting, and in Section \ref{s:hypothesis_test}, we extend the framework to the multiple trial, non-i.i.d. setting and show how to set up an appropriate hypothesis test, which is then analyzed. There is also an appendix in which we provide some context for our mathematical model by comparing it to another recent model of hidden variable theories given by Brandenburer and Yanofsky in \cite{brandenburger:2008}.

\section{The Setting And The Mathematical Model}\label{s:setting}

Let us describe the setup of the Bell test experiment, which is depicted in Figure \ref{fig:Bell_Diagram}. A photon source, such as a low-powered laser, is pointed at an object with specific properties, such as a nonlinear crystal, which should generate an entangled pair of photons in the singlet state. Upon arrival, each of these photons is subjected to a measurement by a detector.

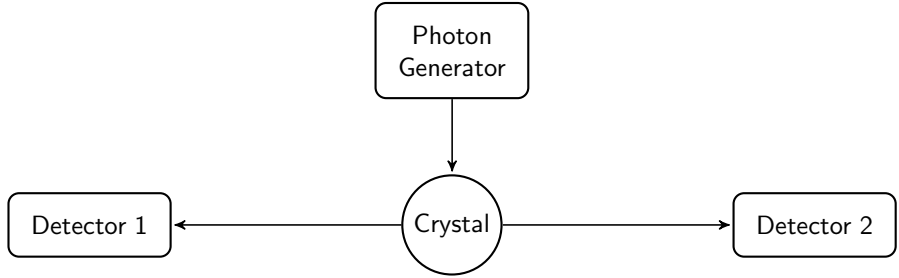
\begin{figure} [h]
\begin{center}
\begin{tikzpicture}[
  font=\sffamily,
  every matrix/.style={ampersand replacement=\&,column sep=2.7cm,row sep=1cm},
  source/.style={draw,thick,rounded corners,inner sep=.3cm},
  process/.style={draw,thick,circle, inner sep = 1mm},
  sink/.style={draw,thick,rounded corners, inner sep=.3cm},
  datastore/.style={draw,very thick,shape=datastore,inner sep=.3cm},
  dots/.style={gray,scale=2},
  to/.style={->,>=stealth',shorten >=1pt,semithick,font=\sffamily\footnotesize},
  every node/.style={align=center}]

\matrix{
\& \node[source] (gen) {Photon\\Generator}; \& \\
\node[sink] (det1) {Detector 1};
\& \node[process] (split) {Crystal};
\& \node[sink] (det2) {Detector 2};\\
};

  \draw[to] (gen) -- node[midway,right]{ } (split);
  \draw[to] (split) -- node[midway,below]{ } (det1);
  \draw[to] (split) -- node[midway,below]{ } (det2);

\end{tikzpicture}
\end{center}
\caption{Diagram of a Bell Test Experiment}
\label{fig:Bell_Diagram}
\end{figure}

As depicted in Figure \ref{fig:Detector_Detail}, Detector 1 has two measurement settings and two possible outputs. The detector measures the polarization of the incoming photons; the setting is the angle at which polarization is measured. The two setting choices, $a$ or $a'$, are chosen to maximize the violation of the CHSH inequality. Detector 2 has a very similar scheme; the only difference is that we label its settings as $b$ and $b'$, to distinguish them from the settings of detector 1. The time of detection of the photons should be calibrated so the selection of the setting choice at detector 1 is spacelike separated from the detection event at detector 2, and vice-versa.

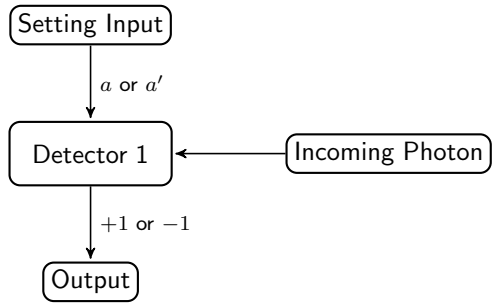
\begin{figure} [h]
\begin{center}
\begin{tikzpicture}[
  font=\sffamily,
  every matrix/.style={ampersand replacement=\&,column sep=1.5cm,row sep=1cm},
  source/.style={draw,thick,rounded corners,inner sep=.1cm},
  process/.style={draw,thick,circle, inner sep = 1mm},
  sink/.style={draw,thick,rounded corners, inner sep=.3cm},
  datastore/.style={draw,very thick,shape=datastore,inner sep=.3cm},
  dots/.style={gray,scale=2},
  to/.style={->,>=stealth',shorten >=1pt,semithick,font=\sffamily\footnotesize},
  every node/.style={align=center}]

\matrix{
\node[source] (input) {Setting Input}; \\
\node[sink] (det1) {Detector 1}; \& \node[source](photon) {Incoming Photon};\\
\node[source] (output) {Output};\\
};

  \draw[to] (input) -- node[midway,right]{$a$ or $a'$} (det1);
  \draw[to] (det1) -- node[midway,right]{$+1$ or $-1$} (output);
  \draw[to] (photon) -- node[midway,right]{ } (det1);

\end{tikzpicture}
\end{center}
\caption{Detail at Detector 1}
\label{fig:Detector_Detail}
\end{figure}

We now model a single trial of the experiment, and leave the repeated-trial scenario to Section \ref{s:hypothesis_test}. The following definition contains the necessary elements for the model. Standard concepts such as ``probability measure" are defined in \cite{chung:1974}.

\begin{definition}\label{def:lambda_A_B_D1_D2}
Let $(\Omega, {\mathcal F}, P)$ be a probability space, where $\Omega$ is a set (the sample space), ${\mathcal F}$ is a $\sigma$-algebra on $\Omega$ (${\mathcal F}$ is a set of events), and $P$ is a probality measure on ${\mathcal F}$. Let $\lambda$, $A$, $B$, $D_1$, $D_2$ be the random functions
$$
\lambda: \Omega \rightarrow \Lambda, \quad \textnormal{$\Lambda$ is a measureable space}, 
$$ $$
A : \Omega \rightarrow \{a,a'\}, \quad B : \Omega \rightarrow \{b,b'\}, 
$$ $$
D_1: \Omega \rightarrow \{-1,1\}, \quad D_2: \Omega \rightarrow \{-1,1\}.
$$
We call $\lambda$ the state of the system prior to measurement. $A$ and $B$ are detector 1's and detector 2's settings, respectively, and $D_1$ and $D_2$ are detector 1's and detector 2's output, respectively.  We label events in $\mathcal F$ corresponding to the outputs of $D_1$ and $D_2$ with the following notation:
\begin{eqnarray*}
+_1 = \{D_1=+1\} \quad\quad -_1 = \{D_1=-1\}\\
   +_2 = \{D_2=+1\} \quad\quad  -_2 = \{D_2=-1\}.
\end{eqnarray*}
\end{definition}

\medskip

The most general of the five random variables above is $\lambda$. This generality is fitting, because $\lambda$ describes the portion of the experiment that we don't directly observe: the state of the photon pair that is theorized to be travelling towards the detectors. Quantum Mechanics has a well-defined description of $\lambda$ and how it triggers the detectors. But we also want to be able to model any conceivable LHVT, so we define the state of the system, $\lambda$, with complete generality.

The other four random variables are more straightforward because they model aspects of the experiment that we can directly observe. We model the detector settings as random variables, as we want the experimenter to toggle the detector settings randomly and independently of anything else going on in the experiment, with the choice  of setting occurring just before the detection event. 

The following three assumptions encapsulate a set of requirements that an experimenter can satisfy in order to properly test Bell's theorem. The notation ``$X \ci Y$"  means ``$X$ is independent of $Y$."

\medskip

\noindent {\sc Experimental Assumption 1:}
\begin{equation}\label{e:exp1}
    A \ci  B.
\end{equation}

\medskip

\noindent The assumption above asserts that choice of measurement settings are independent of each other. In practice, this could be achieved by toggling the measurement setting according to the output of a random number generator attached to the detector, or the output of some independent quantum process that generates randomness, or any other desired source of randomness believed to be uncorrelated with other parts of the experiment.

\medskip

\noindent {\sc Experimental Assumption 2:}
\begin{equation}\label{e:exp2}
    P(A =a)P(A = a')P(B =b)P(B = b') >0.
\end{equation}

\medskip

\noindent The above assumption ensures that the experimenter sets a positive probability of choosing any given detector setting.

\medskip

\noindent {\sc Experimental Assumption 3:}
\begin{equation}\label{e:exp3}
 A \ci \lambda, \quad B \ci\lambda.
\end{equation}

\medskip

\noindent The third and final experimental assumption captures the notion that whatever process is used to choose the detector settings, it should be chosen independently of the state of the approaching particles. Again, we are trusting that our source of randomness for the detector settings is uncorrelated to anything else in the experiment.

(\ref{e:exp3}) is closely related to the ``$\lambda$-independence" assumption that appears in Brandenburger and Yanofsky \cite{brandenburger:2008}. Note that, unlike \cite{brandenburger:2008}, we don't make a slightly stronger assumption that the \emph{joint} distribution of $A$ and $B$ is independent from $\lambda$, written $(A, B) \ci \lambda$; this stronger assumption turns out to be unnecessary in our framework. This contrast is explored in the appendix.

\medskip

We now have a mathematical model for the experiment.  This gives us a framework to discuss a local theory and the conditions it must satisfy. A LHVT must satisfy a locality condition, in addition to the three assumptions described above. To formulate it in a concise manner, we define the random vectors 
$$
V_1 = (D_1, A), \quad \quad V_2 = (D_2, B).
$$
Then the locality assumption is that $V_1$ and $V_2$ are conditionally independent given $\lambda$, which is written as follows:

\medskip

\noindent {\sc Locality Assumption:}
\begin{equation}\label{e:locality}
 (V_1 \ci V_2) |\lambda.
\end{equation}

\medskip

\noindent So, why do we choose equation (\ref{e:locality}) as an expression of locality? Remember that $\lambda$ represents what is going on between the two detectors, prior to the measurement event, that can affect the detection events. Once we condition on this knowledge, what occurs at detector 1 should be independent of what occurs at detector 2. Equation (\ref{e:locality}) says that the events at detector 1 cannot be correlated with the events at detector 2 \emph{beyond} the effects of the shared history of what happened between them prior to detection, represented by $\lambda$. 

Since $V_1$ and $V_2$ each only have four possible outputs, (\ref{e:locality}) is essentially a statement of the conditional independence of a collection of sixteen pairs of events.  For instance, one of the sixteen consequences of (\ref{e:locality}) would be
$$
P\bigg [\big(+_1\cap\{A=a'\}\big)\bigcap \big(-_2\cap\{B=b\}\big)\big|\lambda\bigg ]
$$
\begin{equation}\label{e:oneconseq}
= P\big (+_1\cap\{A=a'\}\big|\lambda\big)\cdot P\big(-_2\cap\{B=b\}\big|\lambda\big).
\end{equation}
When we use (\ref{e:locality}), it will be through equivalences like the one above.

The conditional probabilities in (\ref{e:oneconseq}) are themselves random variables, as defined in \cite{chung:1974}. Theoretically these can be complicated constructions, but if $\lambda$ is a discrete random variable, the dependent random variable $P(E|\lambda)$ is also discrete, taking the value $P(E|\{\lambda=x\})$ when $\lambda = x$. This simplified situation has the benefit of being highly intuitive, and it is explored in the appendix. For now, we make no such simplifying assumption about $\lambda$. 

We refer to a Bell experiment satisfying (\ref{e:exp1}), (\ref{e:exp2}), (\ref{e:exp3}) and (\ref{e:locality}) as being governed by a LHVT. Experimental results inconsistent with assumptions (\ref{e:exp1})-(\ref{e:locality}) can be considered violations of the LHVT hypothesis, implying that one of the assumptions must not hold. We will further explore how to interpret a violation in the conclusion.

\section{Deriving the CHSH Inequality}\label{s:math_develop}

In this section, we work in the fully general setting of Section \ref{s:setting} and derive the CHSH inequality, given in (\ref{e:gCHSH}). Thus our first task is to precisely define the expressions in (\ref{e:gCHSH}). If we condition on the event that detector 1 is set to ``$a$" and detector 2 is set to ``$b$", we can discuss the quantity
\begin{equation}\label{e:Eab(D1D2)}
E_{ab}(D_1D_2):=E(D_1D_2 | A=a,B=b),
\end{equation}
where $E$ denotes the expectation value. It will save space to use shorthands such as ``$a$" or ``$ab'$" for the events $\{A=a\}$ or $\{A=a\}\cap\{B=b'\}$, etc., as is done in (\ref{e:gCHSH}).

Deriving (\ref{e:gCHSH}) in the general setting of Section \ref{s:setting} takes some work. The following notation will be useful:
\begin{equation}\label{e:condprobshorthand}
\mu_{X}(Y|\xi):= {1 \over P(X)}P(Y\cap X|\xi).
\end{equation}
\noindent(\ref{e:condprobshorthand}) is introduced to approximate an intuitive notion of the probability of event $Y$ that is conditioned simultaneously on the event $X$ \emph{and} the random variable $\xi$.   Using this shorthand, we can derive the following expression,
\begin{equation}\label{e:expectab}
E_{ab}(D_1 D_2) = \int_{\Omega} \big[ \mu_{ab}(D_1D_2=+1| \lambda) - \mu_{ab}(D_1 D_2 = -1| \lambda) \big] dP,
\end{equation}
where the integral is taken over $\Omega$ with respect to the probability measure $P$. The justification of equation (\ref{e:expectab}) is given by the following lemma.  Note that the proof makes no use of the locality assumption (\ref{e:locality}).

\begin{lemma}\label{l:Eab(D1_D2)}
Let $a$, $b$, $D_1$, $D_2$ be as in Definition \ref{def:lambda_A_B_D1_D2}. Then, under (\ref{e:exp1}) and (\ref{e:exp2}), the equation (\ref{e:expectab}) holds.
\end{lemma}
\begin{proof}
By  (\ref{e:exp1}) and (\ref{e:exp2}), $P(a\cap b) >0$. If we let $I_{ab}$ denote the indicator function of the event $a \cap b$, we can write
\begin{equation}\label{e:gustavo}
E_{ab}(D_1 D_2) = {{E(I_{ab}D_1 D_2)}\over P(ab)} = {E(E(I_{ab}D_1 D_2|\lambda))\over P(ab)}
\end{equation}
by the definition of conditional expectation (when we condition on events) and the law of iterated expectation. Note that we can be sure that the conditional expectation $E(I_{ab}D_1 D_2|\lambda)$ is guaranteed to exist, as $E(|I_{ab}D_1 D_2|)$ is  finite.  

We claim that 
\begin{equation}\label{e:named}
E(I_{ab}D_1 D_2|\lambda) = E(I_{\{D_1 D_2 = +1\} \cap ab }|\lambda) - E(I_{\{D_1 D_2 = -1\} \cap ab }|\lambda), \quad \textnormal{a.s.}
\end{equation}
To prove the assertion, we must show that for all $A \in \sigma ( \lambda)$,
$$
\int_A  E(I_{\{D_1 D_2 = +1\} \cap ab }|\lambda) - E(I_{\{D_1 D_2 = -1\} \cap ab }|\lambda)dP = \int_A I_{ab} D_1 D_2 dP.
$$
Indeed,
\begin{eqnarray*}
&& \int_A  E(I_{\{D_1 D_2 = +1\} \cap ab }|\lambda) -  E(I_{\{D_1  D_2 = -1\} \cap ab }|\lambda) dP\\
&=& \int_A I_{\{D_1  D_2 = +1\} \cap ab } - I_{\{D_1 D_2 = -1\} \cap ab } dP \\
&=& (+1)P(\{D_1  D_2 = +1\} \cap ab \cap A) + (-1) P(\{D_1  D_2 = -1\} \cap ab \cap A) \\
&=& \int_{\Omega} (D_1 D_2)  I_A I_{ab} dP = \int_A  I_{ab} D_1 D_2 dP,
\end{eqnarray*}
which proves (\ref{e:named}). Plugging (\ref{e:named}) into (\ref{e:gustavo}), we can write
\begin{equation}\label{e:oneoff}
E_{ab}(D_1 D_2)  =   E\bigg{(}{1\over {P(ab)}}\big{\lbrack} E(I_{\{D_1 D_2 = +1\} \cap ab }|\lambda) - E(I_{\{D_1  D_2 = -1\} \cap ab }|\lambda)\big{\rbrack} \bigg{)}.
\end{equation}
Using the notation introduced in (\ref{e:condprobshorthand}), we have
$$
\mu_{ab}(D_1 D_2=+1| \lambda) = {1\over {P(ab)}} E(I_{\{D_1 D_2=+1\} \cap ab}| \lambda),
$$
so we can rewrite (\ref{e:oneoff}) as
$$
E\bigg{(} \mu_{ab}(D_1 D_2=+1| \lambda) - \mu_{ab}(D_1 D_2=-1| \lambda)\bigg).
$$
Thus, (\ref{e:expectab}) holds.
\end{proof}

\vspace{-.18 in} 

$\hfill\Box$

\medskip

\medskip

\medskip

As we work toward the CHSH inequality, it will be useful to expand expression (\ref{e:Eab(D1D2)}), for which we introduce a shorthand for readability:
\begin{eqnarray*}
{\bf a}= \big{[}\mu_a(+_1|\lambda) -\mu_a(-_1|\lambda)\big{]}, \quad {\bf b} = \big{[}\mu_{b}(+_2|\lambda) -\mu_{b}(-_2|\lambda)\big{]}, \\
{\bf b'} = \big{[}\mu_{b'}(+_2|\lambda) -\mu_{b'}(-_2|\lambda)\big{]}, \quad {\bf a'} = \big{[}\mu_{a'}(+_1|\lambda) -\mu_{a'}(-_1|\lambda)\big{]}.
\end{eqnarray*}
The following lemma will also be useful:

\begin{lemma}\label{l:add_cond_probs}
Let $(\Omega, {\mathcal F}, P)$ be a probability space, let $\mathcal  G \subseteq {\mathcal F}$ be a sub-$\sigma$-algebra of $\mathcal F$, and let $\{B_i\}_{i\in I}$ be a countable indexed set of pairwise-disjoint events in $\mathcal F$.  Then,
$$
P(\cup_{i\in I} B_i | \mathcal G) = \sum_{i \in I} P(B_i |\mathcal G), \quad\textnormal{almost surely.}
$$
\end{lemma}

\begin{proof}
This follows in a straightforward manner from the measure-theoretic definition of conditional probability. $\hfill\Box$
\end{proof}

\begin{proposition}\label{p:Eab(D1D2)}
Let $a$, $b$, $D_1$, $D_2$ be as in Definition \ref{def:lambda_A_B_D1_D2}. Then, under (\ref{e:exp1}), (\ref{e:exp2}), (\ref{e:locality}),
\begin{equation}\label{e:cpct}
E_{ab}(D_1D_2) = \int_{\Omega} {\bf a} {\bf b} \hspace{1mm} dP.
\end{equation}
\end{proposition}
\begin{proof}
Note that
\begin{eqnarray*}
\{D_1 D_2 = 1 \} &=& ( +_1 \cap +_2 )\cup ( -_1 \cap -_2 ), \\
\{D_1 D_2 = -1 \} &=& ( +_1 \cap -_2 )\cup ( -_1 \cap +_2 ).
\end{eqnarray*}
Lemma \ref{l:add_cond_probs} can thus be applied to rewrite (\ref{e:expectab}) as
\begin{equation}\label{e:expanded}
\hspace{3mm}\int_{\Omega}\mu_{ab}(+_1\cap +_2| \lambda) + \mu_{ab}(-_1\cap -_2|\lambda) - \big[\mu_{ab}(+_1\cap -_2| \lambda) + \mu_{ab}(-_1\cap +_2| \lambda)\big] dP.
\end{equation}
We now appeal to the locality assumption.  Applying (\ref{e:locality}), as well as (\ref{e:exp1}), we can modify the terms in the integrand in the following way:
\begin{eqnarray*}
\mu_{ab}(+_1 \cap +_2| \lambda) &=& {1\over P(ab)} P(ab\cap +_1 \cap +_2 | \lambda)\\
 &=& {1\over P(a)P(b)} P(a\cap +_1 | \lambda) P(b\cap +_2 |\lambda ) \\
&=& \mu_{a}(+_1| \lambda)  \mu_{b}(+_2|\lambda).
\end{eqnarray*}
Doing the same thing for the three other terms in  (\ref{e:expanded}), we get
\begin{eqnarray*}
E_{ab}(D_1D_2) = \int_{\Omega} \mu_{a}(+_1| \lambda)\mu_b (+_2|\lambda) + \mu_{a}(-_1| \lambda)\mu_b (-_2|\lambda) \\
 -\big[ \mu_{a}(+_1| \lambda)\mu_b (-_2|\lambda) + \mu_{a}(-_1| \lambda)\mu_b (+_2|\lambda)\big] dP
\end{eqnarray*}
\begin{equation}\label{e:compact}
=\int_{\Omega} \big{[}\mu_a(+_1|\lambda) -\mu_a(-_1|\lambda)\big{]}\big{[}\mu_b(+_2|\lambda) -\mu_b(-_2|\lambda)\big{]} dP.
\end{equation}
Thus, (\ref{e:cpct}) holds.
\end{proof}

\vspace{-.18 in} 

$\hfill\Box$

\medskip

\medskip

\medskip

Now, consider the following constant:
\begin{equation}\label{e:CHSH}
 K^{CHSH} := E_{ab}(D_1 D_2) - E_{a'b}(D_1 D_2) + E_{ab'}(D_1 D_2) + E_{a'b'}(D_1 D_2).
\end{equation}
In the local setting, we can calculate bounds that $K^{CHSH}$ must obey. These bounds -- the CHSH inequality (\ref{e:gCHSH}) -- are developed in the next proposition, which requires the following lemma.

\begin{lemma}\label{l:probsbounded}
Let $X$ be an event for which $P(X)>0$ and $X\ci\xi$, where $\xi$ is a random variable. Then for any event $Y$,
\begin{equation}\label{e:probsboundedone}
0\le\mu_X(Y|\xi) \le 1, \quad \textnormal{almost surely.}
\end{equation}
\end{lemma}

\begin{proof}
We have
$$
\mu_X(Y|\xi) = {1\over P(X)} P(Y\cap X | \xi ) = {P(Y\cap X | \xi ) \over P(X|\xi)},
$$
where $P(X|\xi) = P(X)$ holds because $X\ci \xi$. Since $P(X|\xi)\ge P(Y\cap X|\xi)$ and $P(X|\xi)$ is positive (almost surely), we have
$$
\mu_X(Y|\xi) = {P(X\cap Y| \xi) \over P(X|\xi)}\le 1\quad \textnormal{a.s.}
$$
This proves the upper bound in (\ref{e:probsboundedone}). The lower bound holds because a conditional probability such as $P(Y\cap X | \xi )$ is in general greater than or equal to zero almost surely.
\end{proof}

\vspace{-.18 in} 

$\hfill\Box$

\medskip

\medskip

\begin{example}
Under (\ref{e:exp2}) and (\ref{e:exp3}), Lemma \ref{l:probsbounded} applies to expressions such as $\mu_a(+_1|\lambda)$, $\mu_{b'}(-_2|\lambda)$, etc.
\end{example}

\begin{proposition}\label{p:|CHSH|=<2}
(CHSH Inequality)  Let $a$, $b$, $D_1$, $D_2$ be as in Definition \ref{def:lambda_A_B_D1_D2}. Then, under (\ref{e:exp1}), (\ref{e:exp2}), (\ref{e:exp3}), and (\ref{e:locality}),
\begin{equation}\label{e:|CHSH|=<2}
|K^{CHSH}| \leq 2.
\end{equation}
\end{proposition}
\begin{proof}
By Proposition \ref{p:Eab(D1D2)}, we have
\begin{eqnarray*}
K^{CHSH} &=& \int_{\Omega} {\bf a}{\bf b} dP - \int_{\Omega} {\bf a} {\bf b'} dP + \int_{\Omega} {\bf a'}{\bf b} dP + \int_{\Omega} {\bf a'} {\bf b'} dP \\
&=& \int_{\Omega} ({\bf a}  {\bf b}) - ({\bf a} {\bf b'}) + ({\bf a'}  {\bf b}) + ({\bf a'}  {\bf b'} ) dP  = \int_{\Omega} ({\bf a}+ {\bf a'}){\bf b} + ({\bf a'}-{\bf a}){\bf b'} dP.
\end{eqnarray*}
By Lemma \ref{l:probsbounded}, (\ref{e:exp2}) and (\ref{e:exp3}) tell us that ${\bf a}$ and ${\bf a'}$ must lie in the interval $[-1, +1]$. Then by arithmetical considerations, it follows that
\begin{equation}\label{e:arith}
|{\bf a} + {\bf a'}| + |{\bf a'} - {\bf a}| \leq 2 .
\end{equation}
Since $|{\bf b}|,|{\bf b'}| \leq 1$, by (\ref{e:arith}) we have
\begin{eqnarray*}
|K^{CHSH}| &=& \bigg{|}\int_{\Omega} ({\bf a}+ {\bf a'}){\bf b} + ({\bf a'}-{\bf a}){\bf b'} dP \bigg{|} \\
&\leq& \int_{\Omega} |{\bf a}+ {\bf a'}||{\bf b}| + |{\bf a'}-{\bf a}||{\bf b'}| dP \\
&\leq& \int_{\Omega} |{\bf a}+ {\bf a'}| + |{\bf a'}-{\bf a}| dP \le 2.
\end{eqnarray*}

\end{proof}
\vspace{-.45 in} 

$\hfill\Box$

\medskip

\medskip

\medskip

As a consequence of Proposition \ref{p:|CHSH|=<2}, in any LHVT, the quantity $ K^{CHSH}$ must satisfy the simple inequality (\ref{e:|CHSH|=<2}). On the other hand, Quantum Mechanics predicts $K^{CHSH} = 2 \sqrt{2} > 2$. If we repeat the Bell test experiment many times and assume that the results of repeated trials are independent and identically distributed, we can calculate the $K^{CHSH}$ quantity empirically and draw an appropriate conclusion about the theory describing the experiment.

However, as earlier noted, the assumptions of a LHVT do not require repeated trials to be independent and identically distributed, and so we have no reason to assert that the relative frequencies of various outcomes will converge to some underlying probability. \emph{A priori}, we cannot even rule out the (pathological) possibility that each successive trial individually obeys the CHSH inequality (as required by a LHVT), but that the relative frequencies over many trials converge to the quantum values! In the next section, we address this problem.

\section{A Hypothesis Test When Trials Are Not Independent}\label{s:hypothesis_test}

If we run the experiment one time, we will randomly select one particular setting result for $A$ and $B$, and we will observe $D_1D_2$ equal to $+1$ or $-1$.  This one result tells us nothing about the satisfaction or violation of (\ref{e:|CHSH|=<2}).  We must run the experiment many times to discern a pattern. 

Luckily, we can perform a cogent hypothesis test, even without the assumption of independent, identically distributed trials.  Here is a useful analogy that will illustrate how we do this.  Suppose we were to flip 10,000 different coins, and 80\% of them were to come up ``heads."  Then we could reasonably conclude that at least \emph{some} of the 10,000 coins were biased towards heads.  The coins needn't be identically distributed - indeed, perhaps some of the coins were fair - but it is intuitively clear that \emph{some} of them must have been biased.

Analogously, each trial of the Bell test is like a coin flip, resulting in the product $D_1D_2$ being equal to  $+1$ or $-1$. In the previous section, we showed that the assumption of a LHVT puts certain constraints on the probabilities of getting $+1$ or $-1$. If the universe is governed by a LHVT, then the constraint must be satisfied on  \emph{every} trial.  On the other hand, if Quantum Mechanics is obeyed, the constraint is \emph{violated} on every trial. Then, thinking of the analogy, the locality assumption is like the assumption that \emph{every one} of the 10,000 coins are fair, whereas agreement with Quantum Mechanics will predict getting 80\% heads.  The Bell test is of course a little more complicated than coin tossing, but the analogy is a good idea to keep in mind as we design the hypothesis test.

To represent repeated trials, we must extend the framework of Section \ref{s:setting}. Let us define a sequence of random vectors:
\begin{equation}\label{e:ivariables}
\{D_{1i}, D_{2i}, A_i, B_i, \lambda _i\}_{i\in \mathbb N ^+}
\end{equation}
For each $i$, we take the above to be as defined in Definition \ref{def:lambda_A_B_D1_D2}, satisfying conditions (\ref{e:exp1})  and (\ref{e:exp3}), and a strengthened version of (\ref{e:exp2}) . That is, we assume:

\medskip

\noindent {\sc Experimental Assumption 1:}
\begin{equation}\label{e:AiBiindep}
    \forall i, \quad A_i \ci  B_i.
\end{equation}

\medskip

\noindent {\sc Experimental Assumption 2*:}
\begin{equation}\label{e:ABhalf}
   \forall i,  \quad P(A_i = a) = P(B_i = b) = 1/2.
\end{equation}

\medskip

\begin{remark}
(\ref{e:ABhalf}) can be satisfied by appropriate calibration of the experimental apparatus.  Earlier, we assumed only that these probabilities were positive; to prove an analogue of the CHSH inequality that holds over repeated trials, it is useful to assume that all the setting probabilities are calibrated to 1/2.   
\end{remark}

\medskip

\noindent {\sc Experimental Assumption 3:}

\begin{equation}\label{e:irandomchoice}
 A_i \ci \lambda_i, \quad B_i \ci\lambda_i .
\end{equation}

\medskip

An additional point about the $\lambda_i$ needs to be made.  Since $\lambda_i$ models the state of the system at the $i$th trial, the previous $i-1$ trials have already taken place.  Hence, the outcomes of previous trials are in the ``history", and can contribute to or influence the present state of the system.  Mathematically, this is modeled by assuming that the results of previous trials are events in $\sigma (\lambda_i)$.  This yields a \emph{filtration} - i.e., a sequence of nested $\sigma$-algebras:
$$
\textnormal{For } i<j, \quad\sigma(\lambda_i) \subseteq \sigma(\lambda_j).
$$
The filtration is a standard mathematical tool for modeling a time-indexed stochastic process.  The above equation is not used in our argument, but we will need to use the fact that the outcomes of previous trials are events in $\sigma (\lambda_i)$.  The following assumption formalizes this.

\medskip

\noindent {\sc Time Sequentiality:}
For any positive integer $n \geq 2$, let $I$ be a subset of $\{1, 2, ..., n-1\}$ whose cardinality we denote with the letter $m$.  Let $\vec v_1$, $\vec v_2$ be elements of $\{-1, +1\}^m$, let  $\vec w_a$ be an element of $\{a, a'\}^m$, and let $\vec w_b$ be an element of $\{b, b'\}^m$.  Then the following event is in $\sigma ( \lambda_n)$:
\begin{equation}\label{e:time_sequentiality}
\bigcap_{i\in I}\big [ \{D_{1i} = \vec v_{1i} \} \cap \{D_{2i} = \vec v_{2i} \} \cap\{A_{i} = \vec w_{ai} \} \cap \{B_{i} = \vec w_{bi} \} \big ]. 
\end{equation}

\medskip 

The astute reader will notice that this is mathematically equivalent to saying that all the single events such as $\{D_{1i} = \vec v_{1i} \}$ or $\{A_{1j} = \vec w_{1j} \}$, etc. are individually in $\sigma (\lambda_n)$; (\ref{e:time_sequentiality}) is written to emphasize simultaneity of the four events with the same  $i$-index. The significance of asserting that (\ref{e:time_sequentiality}) is in $\sigma (\lambda_n )$ is to encode the notion that $\lambda_n$ can potentially depend on the outcomes of previous trials. Of course, this assumption doesn't require that $\lambda_n$ definitely does have some relation to the outcome of previous trials; $\lambda_n$ could still be independent of this information.

For the final step, we establish a locality assumption corresponding to (\ref{e:locality}):

\medskip

\noindent {\sc Locality Assumption:}
Let $V_{i1} = (D_{i1}, A_i)$ and $V_{i2} = (D_{i2}, B_i)$.  Then
\begin{equation}\label{e:locality1}
 (V_{i1} \ci V_{i2}) |\lambda _i.
\end{equation}

\medskip

This completes the set of assumptions. Now, to formulate the hypothesis test, it will be convenient to define a random variable $C_i$ as a function of the random variables in (\ref{e:ivariables}).  So, let

$$
C_i = \begin{cases} D_{1i}D_{2i}, & \text{if $(A_i, B_i) \ne (a', b)$, } \\  -D_{1i}D_{2i}, & \text{if $(A_i, B_i) = (a', b)$.} \end{cases}
$$
$C_i$ distils the result of the $i$th trial into a single, two-output random variable. As we will see in the next proposition, the CHSH inequality applies to $C_i$ to cap the probability that $C_i = +1$ at 75\%, if we make all of the experimental assumptions plus locality.

\begin{proposition}
Under assumptions (\ref{e:AiBiindep}), (\ref{e:ABhalf}), (\ref{e:irandomchoice}), and the locality assumption (\ref{e:locality1}), we have $P(C_i = +1)\le {3\over 4}$, or equivalently, 
\begin{equation}\label{e:onehalf}
E(C_i) \leq 1/2.
\end{equation}
\end{proposition}
\begin{proof}
By the Law of Iterated Expectations,
$$
E(C_i)= E\big[E(C_i|(A_i,B_i))\big].
$$
Notice that $E(C_i|(A_i,B_i))$ is a discrete random variable with four outputs, corresponding to the four outputs of $(A_i, B_i)$.  Applying (\ref{e:ABhalf}) and (\ref{e:AiBiindep}), we have
$$
E\big[E(C_i|(A_i,B_i))\big]  
$$ $$
={1\over 4} \bigg ( E[C_i|(A_i,B_i)= (a, b)] + E[C_i|(A_i,B_i)= (a', b)] 
$$ $$
+ E[C_i|(A_i,B_i)= (a, b')] + E[C_i|(A_i,B_i)= (a', b')]\bigg ) 
$$ $$
={1\over 4} \bigg ( E_{ab}(D_{1i}D_{2i}) -  E_{a'b}(D_{1i}D_{2i}) +  E_{ab'}(D_{1i}D_{2i}) +  E_{a'b'}(D_{1i}D_{2i})\bigg).
$$
Noting the similarity to (\ref{e:CHSH}), we obtain the following,
\begin{equation}\label{e:CHSHi}
E(C_i) = {1\over 4}K^{CHSH}_i
\end{equation}
where we take $K^{CHSH}_i$ to be as defined in (\ref{e:CHSH}), after replacing the variables with the $i$-indexed versions given in (\ref{e:ivariables}). Assumptions (\ref{e:AiBiindep}), (\ref{e:ABhalf}), (\ref{e:irandomchoice}), and (\ref{e:locality1}) are equivalent to the assumptions of Proposition \ref{p:|CHSH|=<2} if applied to $K^{CHSH}_i$, so the proposition holds. 
\end{proof}

\vspace{-.18 in} 

$\hfill\Box$

\medskip

\medskip

\medskip

For each $i$, $C_i$ is a Bernoulli trial, taking outputs in the set $\{+1, -1\}$, so let us define
$$
p_i := P(C_i = +1).
$$
It is straightforward to compute
\begin{equation}\label{e:likebernoulli}
E(C_i) = 2p_i -1, \quad \textnormal{Var}(C_i) = 4p_i(1-p_i).
\end{equation} 
Under a LHVT, (\ref{e:onehalf}) and (\ref{e:likebernoulli}) implies that $p_i$ must be at most 75\%. On the other hand, quantum mechanics predicts that $E(C_i) = {\sqrt 2 \over 2}$, which yields a $p_i$ of roughly 85.4\%.  This will allow us to discern a difference over many trials.

We can now formulate the hypothesis test in mathematical terms:
$$
H_0:\forall i,  p_i \leq {3\over 4} \quad\quad\quad\quad\quad\quad \textnormal{(LHVT; (\ref{e:AiBiindep})-(\ref{e:locality1}) satisfied)}
$$
$$
H_A: \forall i, p_i = {1 + \sqrt 2 \over 2\sqrt 2} = .854... \quad\quad\quad\quad\quad\quad \textnormal{(Quantum)}
$$
Over $n$ trials, the natural choice for a sample statistic is $\overline{C_n}$, defined as follows:
$$
\overline{C_n} = {\sum_{i=1}^n C_i\over n}.
$$
and so under the assumption of $H_0$, we expect the sample statistic $\overline{C_n}$ to satisfy
\begin{equation}\label{e:nulltest}
E(\overline{C_n}) \leq 1/2.
\end{equation}
We will reject the null hypothesis in favor of the alternative hypothesis if $\overline{C_n} >  z$, where $z$ will be some cut-point exceeding $1/2$ by a little bit.

Let $p_n(\_)$ denote a probability mass function for the first $n$ outputs of $C_i$, and let $\Theta_0$ be the collection of $p_n(\_)$ that satisfy the assumptions  (\ref{e:AiBiindep}) -- (\ref{e:locality1}). $\Theta_0$ thus denotes the collection of allowable distributions under the null hypothesis.  Then the significance level of the hypothesis test -- the probability of Type I error -- is defined to be 
\begin{equation}\label{e:siglevel}
\alpha = \sup_{p_n(\_)\in\Theta_0}P[\overline{C_n} >  z | p_n(\_)].
\end{equation}
Calculating $\alpha$ is somewhat involved.  This is because the null hypothesis does not assert that the various $C_i$ are i.i.d., so equation (\ref{e:nulltest}) alone does not provide us with an asymptotic distribution of $\overline{C_n}$. In the absence of the assumption of i.i.d, we cannot rule out trivialities such as 
\begin{equation}\label{e:totaldependence}
C_1=C_2 = \cdots = C_{n-1} = C_n
\end{equation}
 (total dependence), for which we would have $\alpha = p_i$, independent of $n$!

The following lemma rules out possibilities like (\ref{e:totaldependence}), and it will allow us to demonstrate that $\alpha$ decreases as $n$ increases.

\begin{lemma}\label{l:Ciindependence}
Let $\vec v$ be any vector in  $ \{-1, +1\}^{i-1}$ for which $P(C_1,...,C_{i-1}=\nobreak\vec v)$ is positive. Then, under the null hypothesis -- which subsumes assumptions (\ref{e:AiBiindep})-(\ref{e:locality1}) -- we have
\begin{equation}\label{e:probbound}
P\big(C_i=+1 \big | (C_1,...,C_{i-1})=\vec v\big) \le {3\over 4}.
\end{equation}
\end{lemma}
\begin{proof}
Let $\mathcal C$ denote the event $(C_1,...,C_{i-1})=\vec v$.  Let $C_i$ be a shorthand for the event $C_i = +1$.  Then we have
\begin{eqnarray*}
P(C_i|\mathcal C) &=& {P(C_i\cap \mathcal C) \over P(\mathcal C)} = {1\over P(\mathcal C)} E(I_{C_i \cap\mathcal C}) = {1\over P(\mathcal C)} E\big [ E(I_{C_i \cap\mathcal C}|\lambda_i)\big ] \\
&=& {1\over P(\mathcal C)} \int_\Omega E(I_{C_i \cap\mathcal C}|\lambda_i) dP = {1\over P(\mathcal C)} \int_\Omega P(C_i \cap\mathcal C|\lambda_i) dP.
\end{eqnarray*}
In the integral above, we note that $\mathcal C$ is in $\sigma (\lambda_i )$ by the time-sequential nature of the experiment, encapsulated in equation (\ref{e:time_sequentiality}).  This implies that 
$$
\int_\Omega P(C_i \cap\mathcal C|\lambda_i) dP = \int_\Omega P(C_i |\lambda_i)I_{\mathcal C} dP,
$$
which is a consequence of Theorem 9.1.3 in \cite{chung:1974}. So the integral becomes
$$
\int_{\mathcal C} P(C_i|\lambda_i) dP.
$$
Using an $i$-indexed version of the ``$+_1$" notation introduced in Definition \ref{def:lambda_A_B_D1_D2}, we apply Lemma \ref{l:add_cond_probs} to decompose the integrand into the eight constituent sub-events of $C_i$, obtaining
$$
\int_{\mathcal C} \bigg[ P(+_{1i} \cap +_{2i} \cap a_i \cap b_i |\lambda_i) + P(-_{1i} \cap -_{2i} \cap a_i \cap b_i |\lambda_i)
$$
$$
 +P(+_{1i} \cap +_{2i} \cap a_i \cap b'_i |\lambda_i) + P(-_{1i} \cap -_{2i} \cap a_i \cap b'_i |\lambda_i) 
$$
$$
 + P(+_{1i} \cap -_{2i} \cap a'_i \cap b_i |\lambda_i) + P(-_{1i} \cap +_{2i} \cap a'_i \cap b_i |\lambda_i)
$$
\begin{equation}\label{e:longint1}
 +P(+_{1i} \cap +_{2i} \cap a'_i \cap b'_i |\lambda_i) + P(-_{1i} \cap -_{2i} \cap a'_i \cap b'_i |\lambda_i)\bigg] dP.
\end{equation}
We apply (\ref{e:locality1}) to the first term of (\ref{e:longint1}) to get
$$
P(+_{1i} \cap +_{2i} \cap a_i \cap b_i |\lambda_i) =  P(+_{1i}\cap a_i |\lambda_i) P(+_{2i} \cap b_i |\lambda_i),
$$
and multiplying right-hand side above by $P(a_i \cap b_i) / P(a_i \cap b_i)$ yields, via (\ref{e:AiBiindep}) and (\ref{e:ABhalf}),
\begin{eqnarray*}
P(+_{1i} \cap +_{2i} \cap a_i \cap b_i |\lambda_i)&=&P(a_i \cap b_i)\big [ \mu_{a_i}(+_{1i}|\lambda_i) \mu_{b_i}(+_{2i}|\lambda_i)\big ] \\
&=& {1\over 4}\big [ \mu_{a_i}(+_{1i}|\lambda_i) \mu_{b_i}(+_{2i}|\lambda_i)\big ].
\end{eqnarray*}
The other seven terms simplify the same way, so (\ref{e:longint1}) becomes 
$$
{1\over 4} \int_{\mathcal C} \mu_{a_i}(+_{1i}|\lambda_i) \mu_{b_i}(+_{2i}|\lambda_i) + \mu_{a_i}(-_{1i}|\lambda_i) \mu_{b_i}(-_{2i}|\lambda_i)
$$
$$
+\mu_{a_i}(+_{1i}|\lambda_i) \mu_{b'_i}(+_{2i}|\lambda_i) + \mu_{a_i}(-_{1i}|\lambda_i) \mu_{b'_i}(-_{2i}|\lambda_i)
$$
$$
+\mu_{a'_i}(+_{1i}|\lambda_i) \mu_{b_i}(-_{2i}|\lambda_i) + \mu_{a'_i}(-_{1i}|\lambda_i) \mu_{b_i}(+_{2i}|\lambda_i)
$$
\begin{equation}\label{e:longint2}
+\mu_{a'_i}(+_{1i}|\lambda_i) \mu_{b'_i}(+_{2i}|\lambda_i) + \mu_{a'_i}(-_{1i}|\lambda_i) \mu_{b'_i}(-_{2i}|\lambda_i) dP.
\end{equation}
 If we define 
$$
t = \mu_{a_i}(+_{1i}|\lambda_i) \quad s = \mu_{a'_i}(+_{1i}|\lambda_i) \quad u = \mu_{b_i}(+_{2i}|\lambda_i) \quad v = \mu_{b'_i}(+_{2i}|\lambda_i),
$$
we can factor the integrand in (\ref{e:longint2}) and again apply Lemma \ref{l:add_cond_probs} to obtain
\begin{equation}\label{e:longint3}
{1\over 4}\int_{\mathcal C} t[u+v]+(1-t)[(1-u)+(1-v)]+s[v+(1-u)]+(1-s)[u+(1-v)] dP.
\end{equation}
By Lemma \ref{l:probsbounded}, which applies by (\ref{e:ABhalf}) and (\ref{e:irandomchoice}), we have $s$, $t$, $u$, and $v$ in $[0,1]$. With this constraint, a case analysis shows that the integrand in (\ref{e:longint3}) is always bounded by 3.  Returning to the original expression, we now have 
$$
P(C_i|\mathcal C) \le {1\over P(\mathcal C)}\cdot {1\over 4}\int_{\mathcal C} 3 dP = {3\over 4}.
$$
Hence, the claim is true. 
\end{proof}

\vspace{-.18 in} 

$\hfill\Box$

\medskip

\medskip

\medskip

Lemma \ref{l:Ciindependence} allows us to formulate an upper-limit distribution for $\overline{C_n}$, as shown in the following proposition. The result shows us that over many repetitions of the experiment, $C_i$ cannot do any better at accumulating ``$+1$" outcomes than an independent, identically distributed process that has a ${3\over 4}$ chance of success each time (i.e., a Binomial random variable). In light of Lemma \ref{l:Ciindependence}, this may seem intuitive, but the proof does take some effort.

\begin{proposition}\label{p:H_0binomialbound}
For a fixed positive integer $n$, let $B_n$ be the Binomial random variable corresponding to $n$ trials with probability of success $p_B = {3\over 4}$.  Then, under the assumptions of Lemma \ref{l:Ciindependence}, for a fixed $k\in \{0,...,n\}$, and for $i$ ranging between $1$ and $n$, 
\begin{equation}\label{e:binombound}
P(\textnormal{at least } k\textnormal{ of the } C_i \textnormal{ equal } +1) \le P(B_n\ge k).
\end{equation}
\end{proposition}
\begin{proof}
To show this holds for any fixed positive integer $n$, we use mathematical induction.

\medskip

\noindent Case 1:  $n=1$. 

There are two possibilities for $k$: 0 and 1.  For $k=1$,
$$
P(\textnormal{at least } k\textnormal{ of the } C_i \textnormal{ equal } +1) = P(C_1 = +1) \le {3\over 4} = P(B_1\ge 1),
$$
and for $k=0$, 
$$
P(\textnormal{at least } k\textnormal{ of the } C_i \textnormal{ equal } +1) = 1 = P(B_1\ge 0).
$$

\medskip

\noindent Case 2:  Assume the claim is true for $n$, and derive that it is true for $n+1$.

\medskip

Now, $k$ can range from $0$ to $n+1$.  First, let us prove it for $k$ between $1$ and $n$, and later we will prove the boundary cases of $k=0$ and $k=n+1$.  

Introduce a shorthand, 
$$
P_{n, k} (C) := P(\textnormal{for } 1\le i\le n \textnormal{, at least } k\textnormal{ of the } C_i \textnormal{ equal } +1),
$$
$$
P_{n, k} (B) := P(B_n\ge k),
$$
so what we are trying to prove can now be written as $P_{n+1,k}(C) \le P_{n+1,k}(B)$. By conditioning, 
\begin{equation}\label{e:byconditioning}
P_{n+1,k}(C) = P_{n,k}(C) + p_{C_{n+1}}\cdot \big[P_{n,k-1}(C) -P_{n,k}(C)\big],
\end{equation}
where we note that $\big[P_{n,k-1}(C) -P_{n,k}(C)\big]$ is the probability that we have \emph{exactly} $k-1$ successes after $n$ trials, and $p_{C_{n+1}}$ denotes the probability that $C_{n+1}=+1$, given exactly $k-1$ successes after $n$ trials.  As we are temporarily omitting the possibility that $k=n+1$ or $k=0$, it follows that $P_{n,k}(C)$ and $P_{n,k-1}(C)$  are well-defined and included in the scope of the inductive hypothesis.

Let $S$ be the subset of $\{-1, +1\}^n$ consisting of vectors for which exactly $k-1$ of the entries are $+1$ and $\vec v \in S \Rightarrow P\big [ (C_1,...,C_n) = \vec v \big ] >0$.  We have
\begin{multline*}
  p_{C_{n+1}} \big[P_{n,k-1}(C) -P_{n,k}(C)\big] \\
=\sum_{\vec v \in S} P\big[C_{n+1}= +1\big|(C_1, ..., C_n)= \vec v\big]P\big[(C_1, ..., C_n)= \vec v\big] \\
\le \sum_{\vec v \in S} {3\over 4} P\big[(C_1, ..., C_n)= \vec v\big] \\
= {3\over 4} \sum_{\vec v \in S}P\big[(C_1, ..., C_n)= \vec v\big] \\
= {3\over 4} \big[P_{n,k-1}(C) -P_{n,k}(C)\big],
\end{multline*}
where the inequality above follows by Lemma \ref{l:Ciindependence}. From this, (\ref{e:byconditioning}) can be re-written as 
\begin{equation}\label{e:threefourths}
P_{n+1,k}(C)\le {1\over 4} P_{n,k}(C) + {3\over 4} P_{n,k-1}(C).
\end{equation}
By the inductive hypothesis, $P_{n,k}(C)\le P_{n,k}(B)$ and $P_{n,k-1}(C)\le P_{n,k-1}(B)$, and we have 
\begin{eqnarray*}
{1\over 4} P_{n,k}(B) + {3\over 4}P_{n,k-1}(B)  &=& P_{n,k}(B) + {3\over 4} \big[P_{n,k-1}(B) -P_{n,k}(B)\big] \\
&=& P_{n+1,k}(B).
\end{eqnarray*}
Hence, $P_{n+1,k}(C)\le P_{n+1,k}(B)$.

This leaves only the boundary cases unproven.  For $k=0$, we clearly have
$$
P_{n+1,k}(C) = 1 =  P_{n+1, k} (B),
$$
so the inequality holds easily.  For $k=n+1$, we have
$$
P_{n+1, k}(C) = P_{n,k-1}(C)\cdot P(C_{n+1} = +1 | C_i = +1 \textnormal{ for } i = 1,..., n).
$$
As $P_{n,k-1}(C)\le P_{n,k-1}(B)$ by the inductive hypothesis, and as
$$
P(C_{n+1} = +1 | C_i = +1 \textnormal{ for } i = 1,..., n)\le {3\over 4}
$$
by Lemma \ref{l:Ciindependence}, we have
$$
P_{n+1,k}(C) \le {3\over 4} P_{n,k-1}(B) = P_{n+1,k}(B).
$$
\end{proof}

\vspace{-.4 in} 

$\hfill\Box$

\medskip

\medskip

\medskip

So, under the null hypothesis, the probability of getting at least $k$ ``$C_i=+1$" results over the course of $n$ trials is bounded above by the probability of getting at least $k$ ``successes" over the course of $n$ Bernoulli trials with probability of success ${3\over 4}$.  The bound is sharp: the i.i.d. case with $p_i = {3\over 4}$ is allowed (just not \emph{implied}) by assumptions (\ref{e:AiBiindep})-(\ref{e:locality1}). Note that this result directly pertains to the behavior of $\overline{C_n}$, as the event ``$\overline{C_n}>z$" is equivalent to the event  ``at least $k$ of the $C_i$ equal +1", where $k$ is an integer determined by the particular value of $z$.  

With these results, we have 
\begin{equation}\label{e:boundarycase}
   \alpha = \sup_{p_n(\_)\in\Theta_0} P\bigg(\overline{C_n}>z \bigg | p_n(\_) \bigg) = P \bigg(\overline{C_n}>z\bigg| C_i \textnormal{ i.i.d. and }\forall i,  p_i = {3\over 4}\bigg).
\end{equation}
As $\alpha$ is bounded by the i.i.d case, which is achieved at the boundary of the null parameter space, we can now calculate it.

\begin{corollary}\label{c:normaldist}
If $B_{n, \frac{3}{4}}$ is a Binomial random variable of $n$ trials with probability of success $\frac{3}{4}$, then 
\begin{equation}\label{e:binomdist}
\alpha = P\left(B_{n,\frac{3}{4}}> \frac{n}{2}(z+1)\right).
\end{equation}
\end{corollary}

For various particular choices of $n$ and $z$, it may be accurate to estimate $\alpha$ using the asymptotic Normal distribution, especially for large choices of $n$. Then the approximation would be $\alpha  \approx 1 - \Phi {2z -1\overwithdelims () \sqrt{3}/ \sqrt{n}}$, where $\Phi(x)$ is the cumulative distribution function of the Standard Normal distribution. However, care should be used, as the Central Limit Theorem only states that 
\begin{equation}\label{e:stghst}
\frac{\overline{C_n}-\frac{1}{2}}{\sqrt{\frac{3}{4}}/\sqrt{n}}\sim N(0,1),
\end{equation}
which does not directly apply to (\ref{e:binomdist}) for fixed choices of $z$ as $n\to\infty$. It is safest to use the Binomial cumulative distribution function to calculate $\alpha$ exactly.

We can also calculate the power of the test, as the alternative hypothesis $H_A$ specifies the distribution of the $C_i$ exactly.  The probabilities $p_i$ are all equal to $\sim .854$, and the quantum mechanical description of the experiment asserts that successive trials are independent (as is intuitive). The power is $1 - \beta$, where $\beta$ is defined as
\begin{equation}\label{e:beta}
\beta = P(\overline{C_n} \le z | H_A).
\end{equation}
From (\ref{e:likebernoulli}), we calculate $E(C_i) = {\sqrt 2 \over 2}$ and $\textnormal{Var}(C_i) = {1\over 2}$, then
\begin{equation}\label{e:betaapprox}
\beta =P\left(B_{n, .854...}<\frac{n}{2}(z+1)\right).
\end{equation}
This can be calculated exactly, or estimated asymptotically with the Normal distribution, subject to the same caveats discussed in the previous paragraph. The Normal approximation for $\beta$ is $\beta \approx \Phi{2z - \sqrt 2\overwithdelims () \sqrt{2}/ \sqrt{n}}$.

To obtain statistical significance, the needed number of trials is not especially high. If the quantum prediction is correct, then $\overline{C_n}$ should tend to ${\sqrt 2\over 2}$.  Hence, if after $n$ trials, $C_n$ is about ${\sqrt 2\over 2}$, we can calculate a $p$-value, using $z={\sqrt 2\over 2}$ in (\ref{e:binomdist}):
\begin{equation}\label{e:p-value}
\textnormal{p-value } = P\left(B_{n, \frac{3}{4}}>\frac{n(\sqrt 2 + 2)}{4}\right).
\end{equation}
For example, to get a $p$-value of $\alpha < .05$, it would suffice to have $n\ge 50$ trials.  

The p-value calculated in (\ref{e:p-value}) is comparable to the figure claimed by \cite{barrett:2002}, and is not larger than the relevant p-values calculated numerically in \cite{zhang:2011}. The martingale-based analysis of \cite{gill:2003} would result in a larger p-value, as discussed in \cite{zhang:2011}; this is due \cite{gill:2003}'s use of the loose (though computationally simple) Azuma-Hoeffding inequality \cite{hoeffding:1963,azuma:1967} to bound the upper tail probabilities, as opposed to exact figures that can be obtained from the Binomial distribution. Tighter Azuma-Hoeffding bounds can be applied, such as expression (8) in \cite{zhang:2013}, which in our setting simplifies to
\begin{equation*}
\textnormal{p-value}\le \left[\left(\frac{1}{2-\sqrt 2}\right)^{\frac{2-\sqrt 2}{4}}\left(\frac{3}{2+\sqrt 2}\right)^{\frac{2+\sqrt 2}{4}}\right]^n.
\end{equation*}
The above bound is easier to compute than the Binomial cumulative distribution function, but there is still a meaningful gap between the bound and the exact figures:
\begin{center}
\begin{tabular}{|c|cccc|}
\multicolumn{1}{c}{n} & \multicolumn{1}{c}{$10$} & \multicolumn{1}{c}{$10^2$}& \multicolumn{1}{c}{$10^3$}& \multicolumn{1}{c}{$10^4$}\\
\hline
Exact p-value (\ref{e:p-value})  &  .2440  &  .0054  &  $6.34\times10^{-16}$  &  $8.58\times 10^{-142}$    \\
A-H Bound (8) in \cite{zhang:2013} &  .7256  &  .0405  &  $1.18\times10^{-14}$  &  $5.03\times 10^{-140}$    \\
\hline
\end{tabular}
\end{center}
As the table reveals, the difference between the upper bound and the exact calculation is roughly two orders of magnitude for larger values of $n$.

\begin{remark}
To calculate the power of the test, we used our knowledge of the quantum predictions. $H_A$ could be extended to include \emph{any} violation of locality; from a hypothesis test standpoint, our knowledge of the precise quantum predictions is not necessary. Smaller (sub-quantum) violations of the CHSH inequality would take more trials to detect. And violations of the inequality on \emph{some} trials, balanced by trials that obey the inequality, could be statistically undetectable if the trials obeying the inequality were to do so by a large enough margin.
\end{remark}

\section{Conclusion}\label{s:conclusion}

We have shown that the CHSH inequality can be proved in a completely general measure-theoretic framework, and furthermore that a hypothesis test can definitively test locality in an experimental setting.

By working in a precise setting, we gain the benefit of clearly delineating all of the assumptions being made. If $H_A$ is supported by experiment, one of the various assumptions must be false. Under most standard interpretations of the quantum description of a Bell experiment, (\ref{e:exp1})-(\ref{e:exp3}) can be satisfied and it is the locality assumption, (\ref{e:locality1}), that is violated. As Quantum Mechanics is a successful theory upheld by countless experiments, it would be logical to attribute the failure of $H_0$ to a quantum violation of (\ref{e:locality1}).

However, the formulation of $H_0$, and the derivation of the CHSH inequality (\ref{e:|CHSH|=<2}) also rest on four other assumptions; the ``experimental assumptions," (\ref{e:AiBiindep}), (\ref{e:ABhalf}), and (\ref{e:irandomchoice}), and time sequentiality, (\ref{e:time_sequentiality}). A physical theory could violate $H_0$, but still satisfy locality so long as one of the other assumptions turned out not to hold.  

It is not clear that a violation of the time sequentiality assumption (\ref{e:time_sequentiality}) would have any physical interpretation, as (\ref{e:time_sequentiality}) is really a technical detail of how to model the problem -- akin to the more basic assumption that we can model the problem with a probability space and random variables to begin with. As for the two assumptions (\ref{e:AiBiindep}) and (\ref{e:ABhalf}), these can be compared to observed data and confirmed to any desired degree of certainty.\footnote{The reader may note that confirming these two assumptions by appealing to experimental data would require an assumption that the random variable sequences $\{A_i\}$ and $\{B_i\}$ are i.i.d. -- exactly the sort of assumption we are trying to avoid in this paper. However, the difference is this: we \emph{observe} $\{A_i\}$ and $\{B_i\}$, and we may come to a reasonable conclusion that we are observing an i.i.d. sequence, whereas we will never be able to conclude this about the unobserved sequence $\{\lambda_i\}$.} On the other hand, (\ref{e:irandomchoice}) is a different creature. Equation (\ref{e:irandomchoice}) states that two observable random variables, $A_i$ and $B_i$, are independent of an unobservable random variable, $\lambda_i$, and therefore this assumption cannot be directly tested.

What would a violation of (\ref{e:irandomchoice}) imply? This would mean that whatever process you were using to randomly set the detector settings was influenced by the state of the system prior to detection, $\lambda_i$. Since we can choose any source of randomness -- a separate quantum process, a random number generator on a computer, random fluctuations of the cosmic background radiation -- to toggle the detector settings, the state of the system $\lambda_i$ would have to be correlated with all sorts of seemingly unrelated processes. However, this would be the only alternative explanation, if we are to keep the locality assumption. 

Sometimes it is claimed that it is not locality, but \emph{realism} that must be abandoned. However, there is some debate about whether realism is a well-defined, required concept in the context of Bell experiments \cite{gisin:2012}, and there is no clear invocation of realism at any point in this paper (assumption (\ref{e:irandomchoice}) is more aptly referred to as a free-will assumption, and (\ref{e:locality}) is of course a locality assumption). It could be argued that modeling the problem using the usual notions of probability fundamentally presupposes a realist viewpoint, but then it is not clear what a non-realist -- but local -- theory would be, or how such a theory could be modeled. In any case, to claim that the CHSH inequality rests on an assumption of realism requires being able to identify which of the assumptions and/or deductive steps in Sections \ref{s:setting}-\ref{s:hypothesis_test} should be identified with realism.

This paper assumes that every trial results in a detection event at both ends of the laboratory. In practice, however, there are limits in the detection efficiency of real-world particle detectors that result in most photons going undetected, so many trials end with only one detector detecting a photon, or no detections at all: see, for example, \cite{weihs:1998}, where detection efficiency was only 5\%. To properly model a real-world experiment with this constraint, one would have to allow for a third outcome, ``undetected"  or ``0", in addition to the two outcomes ``$+1$" and ``$-1$". Previous papers \cite{pearle:1970,clauserhorne:1974,mermin:1987} have analyzed how to model this additional-outcome experiment and it has been found that, for a CHSH experiment using the singlet state, Quantum Mechanics is distinguishable from any LHVT so long as the detection efficiency exceeds a crucial cut-off of about 83\%, an efficiency that has not yet been achieved in a CHSH experiment. Detection-efficiency issues can be addressed in a completely general measure-theoretic framework without making i.i.d. assumptions about repeated trials; this is done in a separate work \cite{dissertation}.

\newpage

{\bf Acknowledgements}\quad
The author would like to thank Michael Mislove and Keye Martin for their support and guidance, as well as Gustavo Didier and Lev Kaplan for their helpful comments and suggestions.

\begin{appendix}
\section*{Appendix}
\setcounter{section}{1}

In this paper, we worked in the most general measure-theoretic setting. In addition to requiring more work, the general setting can make it harder to gain an intuitive grasp of the probabilistic assumptions in the model. In contrast, the system of Brandenburger and Yanofsky \cite{brandenburger:2008} involves a simplifying assumption that the state of the system, $\lambda$, is a discrete random variable with finitely many outputs. With this assumption, notions such as ``locality" and ``$\lambda$-independence" are easier to formulate and easier to understand. Though the finite-$\lambda$ assumption restricts the type of theories one can model, this is not as egregious as it might seem: in some hidden-variable situations, any possible correlation scenario can be modeled by a finite-output $\lambda$, as discussed in \cite{fritz:2012}. 

In this appendix, we investigate what happens to the system of Section \ref{s:setting} if we make the additional assumption that $\lambda$ is a discrete random variable with finitely many outputs. This will allow us to directly compare our system to the system of \cite{brandenburger:2008}, as well as to illustrate and clarify the nature of our particular choices of assumptions.

Before restricting ourselves to the finite-$\lambda$ setting, we can show that, working in the system of Section \ref{s:setting}, we can derive the following alternate version of (\ref{e:exp3}): 

\medskip

\noindent {\sc Experimental Assumption 3*:}
\begin{equation}\label{e:choicemod}
 (A, B) \ci\lambda.
\end{equation}

\medskip

\noindent The above alternative version of (\ref{e:exp3}) is more similar to the ``$\lambda$-independence" assumption as formulated in \cite{brandenburger:2008}. (\ref{e:choicemod}) is also a stronger assumption than (\ref{e:exp3}): one can verify that (\ref{e:choicemod}) directly implies (\ref{e:exp3}), but the converse does not hold. However,  if we also assume (\ref{e:exp1}) and (\ref{e:locality}), we \emph{can} derive (\ref{e:choicemod}) from (\ref{e:exp3}).

\begin{proposition}\label{p:comparing_lambda_independence}
The condition (\ref{e:exp3}), in conjunction with (\ref{e:exp1}) and (\ref{e:locality}), implies (\ref{e:choicemod}).  
\end{proposition}
\begin{proof}
We show that $P\big( (A, B) = (a, b) | \lambda\big) = P\big( (A, B) = (a, b)\big )$; the proof for the three other cases $(a, b')$, $(a', b)$, and $(a', b')$ is the same.  We have
$$
P\big( (A, B) = (a, b) | \lambda\big) = \sum_{i=\pm 1}\sum_{j=\pm 1}P\big(\{A=a\}\cap \{B=b\} \cap (\{D_1 = i\} \cap \{D_2 = j\})\big |\lambda \big ),
$$
by Lemma \ref{l:add_cond_probs}. Applying (\ref{e:locality}) to the above expression yields 
\begin{equation*}
 \sum_{i=\pm 1}\sum_{j=\pm 1} P\big (\{D_1=i\}\cap\{A=a\}\big|\lambda\big)\cdot P\big(\{D_2=j\}\cap\{B=b\}\big|\lambda\big). 
\end{equation*}
Factoring the expression and applying Lemma \ref{l:add_cond_probs}, we obtain
\begin{eqnarray*}
&&P\big(+_2\cap\{B=b\}\big|\lambda\big)\bigg [ P\big (+_1\cap\{A=a\}\big|\lambda\big) + P\big (-_1\cap\{A=a\}\big|\lambda\big)\bigg ] \\
&&+ P\big(-_2\cap\{B=b\}\big|\lambda\big)\bigg [ P\big (+_1\cap\{A=a\}\big|\lambda\big) + P\big (-_1\cap\{A=a\} \big|\lambda\big)\bigg ] \\
&=& P\big(+_2\cap\{B=b\}\big|\lambda\big)\cdot P(A=a|\lambda) +P\big(-_2\cap\{B=b\}\big|\lambda\big)\cdot P(A=a|\lambda) \\
&=&P(A = a|\lambda)P(B = b|\lambda).
\end{eqnarray*}
\noindent Now, by applying (\ref{e:exp3}) and then (\ref{e:exp1}), we have
$$
P(A = a|\lambda)P(B = b|\lambda) = P(A = a)P(B = b) = P\big ( (A, B) = (a, b) \big ).
$$

\end{proof}

\vspace{-.46 in} 

$\hfill\Box$

\medskip

\medskip

\medskip

Proposition \ref{p:comparing_lambda_independence} rules out the possibility of $\lambda$ having some dependence on the joint distribution of $A$ and $B$. 

Moving forward, we now make the assumption that the random variable $\lambda$, introduced in Section \ref{s:setting}, is of the form
\begin{equation}\label{e:lambda_is_finite}
\lambda : \Omega \rightarrow \Lambda = \{l_1, ..., l_n\},
\end{equation}
so $\Lambda$ is now taken to be a finite set containing $n$ elements.  (The nature of its constituent elements $l_i$ is not characterized, or important.) We also assume that for all $i$, $P(\lambda = l_i)> 0$; any zero-probability event has no observable effect on the behavior of the model, so we remove such events from consideration. Henceforth we will use $l_i$ to refer to the event $\{\lambda = l_i\} = \lambda^{-1}(l_i)\subseteq\Omega$, unless doing so could create ambiguity.  

With the assumption that $\lambda$ is finite, the expression (\ref{e:exp3}) is now equivalent to
\begin{equation*}
\forall i\in \{1, ..., n\}, {\bf a} \in \{a, a'\}, \textnormal{ and }{\bf b} \in \{b, b'\},
\end{equation*}
\begin{equation*}
P(A={\bf a}\cap l_i) =  P(A={\bf a})P(l_i) \quad \textnormal{ and }
\end{equation*}
\begin{equation*}
 P(B={\bf b}\cap l_i) =  P(B={\bf b})P(l_i). 
\end{equation*}
The interpretation of (\ref{e:locality}) is simplified as well. This is because for any event $E$, $P(E|\lambda)$ will now just be a simple function that is equal to $P(E|l_i)$ on each set $l_i$.  Now (\ref{e:locality}) is equivalent to the condition
\begin{equation*}
 \forall i\in \{1,...,n\}, {\bf a} \in \{a, a'\}, {\bf b} \in \{b, b'\}, j_a \in \{+_1, -_1\}, \textnormal{ and }k_b\in \{+_2, -_2\},
\end{equation*}
\begin{equation}\label{e:finite_locality}
P(j_a, {\bf a}, k_b, {\bf b}|l_i) = P(j_a, {\bf a}|l_i)P(k_b, {\bf b}|l_i).
\end{equation}

Note that the conditionals in (\ref{e:finite_locality}) are events, not random variables, resulting in a simpler construction when compared to expressions like (\ref{e:oneconseq}). So now, all of the assumptions (\ref{e:exp1}) - (\ref{e:locality}) can be expressed in terms of elementary probabilistic statements concerning a finite collection of events.

We now show that the axiomatization of Section \ref{s:setting} is essentially equivalent to the axiomatization of \cite{brandenburger:2008}, when applied to the relevant experimental setup (i.e., an experiment with two detectors, two detector settings, and two outcomes). To do this, note that if we assume (\ref{e:exp1}), (\ref{e:exp2}), and replace (\ref{e:exp3}) with the stronger (\ref{e:choicemod}), then the following statement,
\begin{equation}\label{e:locality_equivalence_A}
\forall i, {\bf a}, {\bf b}, j_a, k_b,\quad {P(j_a, {\bf a}, k_b, {\bf b},l_i)\over P(l_i)} = {P(j_a, {\bf a},l_i)P(k_b, {\bf b},l_i)\over P(l_i)^2}
\end{equation}
is equivalent to 
\begin{equation}\label{e:locality_equivalence_B}
\forall i, {\bf a}, {\bf b}, j_a, k_b, \quad {P(j_a, {\bf a}, k_b, {\bf b},l_i)\over P({\bf a}, {\bf b}, l_i)} = {P(j_a, {\bf a},l_i)P(k_b, {\bf b},l_i)\over P({\bf a}, l_i)P({\bf b}, l_i)}.
\end{equation}
\noindent Demonstrating the above biconditional is a straightforward exercise. Note that (\ref{e:locality_equivalence_A}) is equivalent to (\ref{e:finite_locality}). Now, recall that by Proposition \ref{p:comparing_lambda_independence}, we have the following logical relationship between assumptions,
\begin{equation*}
(\ref{e:exp1}), (\ref{e:exp2}), (\ref{e:exp3}), (\ref{e:locality})\quad \Leftrightarrow \quad (\ref{e:exp1}), (\ref{e:exp2}), (\ref{e:choicemod}), (\ref{e:locality}),
\end{equation*}
and in the finite-$\lambda$ setting, we have (\ref{e:locality}) $\Leftrightarrow$ (\ref{e:finite_locality}), so we can say that 
\begin{equation*}
(\ref{e:exp1}), (\ref{e:exp2}), (\ref{e:choicemod}), (\ref{e:locality}) \quad \Leftrightarrow \quad (\ref{e:exp1}), (\ref{e:exp2}), (\ref{e:choicemod}), (\ref{e:locality_equivalence_B}).
\end{equation*}
The collection of assumptions on the right side of the above equivalence is closely related to the framework of \cite{brandenburger:2008} as it would apply to the 2-detector, 2-setting, 2-outcome scenario. (\ref{e:choicemod}) is equivalent to Definition 2.4 (``$\lambda$-independence") in \cite{brandenburger:2008}, and (\ref{e:locality_equivalence_B}) is equivalent to Definition 2.10 (``locality"). So in a finite-$\lambda$ setting, our framework - i.e., the set of conditions (\ref{e:exp1}) - (\ref{e:locality}) - is equivalent to the Brandenburger/Yanofsky framework applied to a 2-dector, 2-setting, 2-outcome scenario where measurement choices are independent from each other (the condition (\ref{e:exp1})) and none of the measurement settings have trivial probabilities (the condition (\ref{e:exp2})). 

\end{appendix}




\bibliographystyle{spphys}       
\bibliography{Bibliography}   

%
%

\end{document}